\documentclass[english,runningheads,10pt]{llncs}

\usepackage{tabularx,booktabs,multirow,delarray,array}
\usepackage{graphicx,amssymb,amsmath,amssymb}
\usepackage{enumerate}
\usepackage{algorithmic}
\usepackage[ruled]{algorithm2e}
\usepackage{wrapfig}
\usepackage{latexsym}
\usepackage{lineno}
\usepackage{xcolor}
\usepackage{enumitem}

\newcommand{\ignore}[1]{ }

\newcounter{rem}
\newtheorem{innercustomgeneric}{\customgenericname}
\providecommand{\customgenericname}{}
\newcommand{\newcustomtheorem}[2]{%
  \newenvironment{#1}[1]
  {%
   \renewcommand\customgenericname{#2}%
   \renewcommand\theinnercustomgeneric{##1}%
   \innercustomgeneric
  }
  {\endinnercustomgeneric}
}

\newcustomtheorem{customthm}{Theorem}
\newcustomtheorem{customlemma}{Lemma}

\def\qed{\hbox{\rlap{$\sqcap$}$\sqcup$}}

\def\calP{\mathcal{P}}

\def\calF{\mathcal{F}}

\def\calS{\mathcal{S}}

\newenvironment{proof}{\par\noindent{\bf Proof:}}{\mbox{}\hfill$\qed$\\}

\begin{document}

\title{Computing a rectilinear shortest path amid splinegons in plane}
\author{
Tameem Choudhury\inst{1}
\and
R. Inkulu\inst{1}\thanks{This research is supported in part by NBHM grant 248(17)2014-R\&D-II/1049}
}

\institute{
Department of Computer Science \& Engineering\\
IIT Guwahati, India\\
\email{\{ctameem,rinkulu\}@iitg.ac.in}
}

\maketitle

\pagenumbering{arabic}
\setcounter{page}{1}

\begin{abstract}
We reduce the problem of computing a rectilinear shortest path between two given points $s$ and $t$ in the splinegonal domain $\calS$ to the problem of computing a rectilinear shortest path between two points in the polygonal domain.
As part of this, we define a polygonal domain $\calP$ from $\calS$ and transform a rectilinear shortest path computed in $\calP$ to a path between $s$ and $t$ amid splinegon obstacles in $\calS$.
When $\calS$ comprises of $h$ pairwise disjoint splinegons with a total of $n$ vertices, excluding the time to compute a rectilinear shortest path amid polygons in $\calP$, our reduction algorithm takes $O(n + h \lg{n})$ time.
For the special case of $\calS$ comprising of concave-in splinegons, we have devised another algorithm in which the reduction procedure does not rely on the structures used in the algorithm to compute a rectilinear shortest path in polygonal domain. 
As part of these, we have characterized few of the properties of rectilinear shortest paths amid splinegons which could be of independent interest.
\end{abstract}

\section{Introduction}
\label{sect:intro}

Computing obstacle avoiding shortest path between two points is both fundamental and well-known in computational geometry. 
The case on polygonal obstacles has been well studied (Ghosh and Mount~\cite{journals/siamcomp/GhoshM91}, Kapoor and Maheshwari~\cite{journals/siamcomp/KapoorM00}, Hershberger and Suri~\cite{journals/siamcomp/HershbergerS99}, Kapoor and Maheshwari~\cite{conf/socg/Kapoor88}, Kapoor et al.~\cite{journals/dcg/KapoorMM97}, Mitchell~\cite{journals/ijcga/Mitchell96}, Rohnert~\cite{journals/ipl/Rohnert86}, Storer and Reif~\cite{journals/jacm/StorerR94}, and Inkulu et al.~\cite{journals/corr/InkuluKM10}).
Algorithms for visibility on which many Euclidean shortest path algorithms rely are detailed in Ghosh~\cite{books/visalgo/skghosh2007}.
Clarkson et al.~\cite{conf/socg/ClarksonKV87}, Inkulu and Kapoor~\cite{journals/comgeo/InkuluK09a} and Chen et al.~\cite{journals/jocg/ChenIW16} devised algorithms to compute a rectilinear shortest path amid polygonal obstacles.
In this paper, we devise an algorithm to compute a rectilinear shortest path amid planar curved obstacles.
In specific, as in Dobkin and Souvaine~\cite{journals/algorithmica/DobkinS90}, Dobkin et al.~\cite{journals/algorithmica/DobkinSW88} and Melissaratos and Souvaine~\cite{journals/siamcomp/MelissaratosS92}, we use splinegons to model planar curved objects.
Chen and Wang~\cite{conf/socg/ChenW13} and Hershberger et al.~\cite{conf/socg/HershbergerSY13} devised algorithms to compute a Euclidean shortest path amid curved obstacles in plane.
To our knowledge, this is the first work to compute an optimal rectilinear shortest path amid splinegons in plane. 
Since splinegons model the real-world obstacles more closely than the simple polygons, all the applications of computing shortest paths extend to computing shortest paths amid splinegons as well.

We first introduce terminology from \cite{journals/algorithmica/DobkinS90,journals/algorithmica/DobkinSW88,journals/siamcomp/MelissaratosS92}.
A {(simple) \it splinegon} $S$ is a simple region formed by replacing each edge $e_i$ of a simple polygon $P$ by a curved edge $s_i$ joining the endpoints of $e_i$ such that the region $S$-$seg_i$ bounded by the curve $s_i$ and the line segment $e_i$ is convex~\cite{journals/algorithmica/DobkinSW88}.
The new edge need not be smooth; a sufficient condition is that there exists a left-hand and a right-hand derivative at each point on the splinegon.
The vertices of $S$ are the vertices of $P$.
The polygon $P$ is called the {\it carrier polygon} of the splinegon $S$.
If $S$-$seg_i \subseteq S$, then we say that the edge $e_i$ is {\it concave-in}.
Otherwise, we say that the edge $e_i$ is {\it concave-out}.
We call a splinegon {\it concave-in splinegon} whenever each is of its edges is concave-in.

We assume that the combinatorial complexity of each splinegon edge is $O(1)$, and each of the primitive operations on a splinegon edge can be performed in $O(1)$ time. 
These operations include computing the points of intersections of a splinegon edge with a line, computing the tangents (if any) between two given splinegon edges, finding the tangents between a point and a splinegon edge, computing the distance between two points along a splinegon edge, and finding a point on splinegon edge that has a horizontal or vertical tangent to that splinegon edge at that point.
We assume that no two carrier polygons intersect, and the carrier polygon of any splinegon $S$ does not intersect with another splinegon $S' \in \calS$.

The input {\it splinegonal domain} $\cal{S}$ comprises of $h$ pairwise disjoint splinegons in $\mathbb{R}^2$.
We let $n$ be the number of vertices that together define all the splinegons in $\cal{S}$.
The {\it free space $\cal{F(S)}$} of a splinegonal domain $\calS$ is defined as the closure of $\mathbb{R}^2$ excluding the union of the interior of splinegon obstacles in $\calS$.
Given a splinegonal domain $\calS$ and two given points $s, t \in \cal{F(S)}$, our algorithm computes a shortest path in rectilinear metric, termed {\it rectilinear shortest path}, between $s$ and $t$ that lie in $\cal{F(S)}$.
We reduce the problem of computing a rectilinear shortest path between $s$ and $t$ amid splinegon obstacles in $\calS$ to the problem of computing a rectilinear shortest path between two points amid obstacles in a polygonal domain $\calP$, where $\calP$ is computed from $\calS$.
We assume points $s$ and $t$ are exterior to carrier polygons of splinegon obstacles in $\calS$.
In specific, we prove that this path is a shortest one with respect to rectilinear metric in $\cal{F(S)}$ between $s$ and $t$. 
Analogous to $\cal{F(S)}$, the {\it free space $\cal{F(P)}$} of a polygonal domain $\calP$ is defined as the closure of $\mathbb{R}^2$ excluding the union of the interior of polygonal obstacles in $\calP$.

Chen and Wang~\cite{journals/talg/ChenW15} extended the corridor structures defined for polygonal domains in \cite{conf/socg/Kapoor88,journals/dcg/KapoorMM97} to splinegonal domains.
In computing corridors, this result used bounded degree decomposition of $\cal{F(S)}$ which is analogous to triangulation of the free space of polygonal domain \cite{journals/ijcga/Bar-YehudaC94}. 
Similar to \cite{journals/comgeo/InkuluK09a}, using the corridor decomposition of $\cal{F(S)}$, we characterize rectilinear shortest paths in splinegonal domains. 
These properties facilitate in computing a connected undirected graph $G_{\calS}$ that contains a rectilinear shortest path between $s$ and $t$.
We compute a set $\calP$ of pairwise disjoint simple polygons in $\mathbb{R}^2$ such that each polygon $P \in \calP$ corresponds to a unique splinegon $S \in \calS$ and the vertices of $P$ lie on the boundary of $S$ while ensuring $s, t \in \cal{F(P)}$.
\cite{journals/comgeo/InkuluK09a} gave a constructive proof showing that there exists a connected undirected graph $G_{\calP}$ that contains a rectilinear shortest path between $s$ and $t$ amid $\calP$. 
In computing $\calP$, we ensure the graph $G_{\calS}$ is precisely same as $G_{\calP}$.
We use the algorithm given in~\cite{journals/comgeo/InkuluK09a} to compute a rectilinear shortest path between $s$ and $t$ in $\cal{F(P)}$, and modify the same to a rectilinear shortest path between $s$ and $t$ in $\cal{F(S)}$.
Hence, the reduction of the problem of computing a rectilinear shortest path between $s$ and $t$ amid splinegon obstacles in $\calS$ to the problem of computing a rectilinear shortest path between two points located in $\cal{F(P)}$ amid polygonal obstacles in $\calP$.
Our reduction procedure takes $O(n + h \lg{n})$ time while excluding the time to compute a rectilinear shortest path amid polygonal obstacles in $\calP$.

Since this reduction assumes that a rectilinear shortest path in $\cal{F(P)}$ is computed using the algorithm from \cite{journals/comgeo/InkuluK09a}, we devise another algorithm which works independent of the structures used in the algorithm to compute a rectilinear shortest path in polygonal domain. 
In specific, this reduction algorithm works whenever the shortest path computed in the polygonal domain is polygonal.
However, this algorithm is applicable when $\calS$ comprises of concave-in splinegon obstacles.
The reduction in this algorithm takes $O(n+(h+k)\lg{n}+(h+k+k')\lg{h+k})$ time, again while excluding the time to compute a rectilinear shortest path amid polygonal obstacles in $\calP$.
Let $R$ be the polygonal rectilinear shortest path between $s$ and $t$ amid polygonal obstacles in $\calP$ that was output by the algorithm used.
Then $k$ is the number of line segments in $P$ and $k'$ is the number of points of intersections of that path with splinegons in $\calS$.

We call the vertices of graph as nodes and the vertices of polygonal/splinegonal domain as vertices.
The polygon corresponding to a splinegon $S$ which is computed in the reduction procedure is denoted with $P_S$.
Note that $P_S$ not necessarily be same as the carrier polygon of $S$. 
For any splinegon $S$ (resp. polygon $P$), the boundary of $S$ (resp. $P$) is denoted with $bd(S)$ (resp. $bd(P)$).
A shortest path between $s$ and $t$ amid splinegons (resp. polygons) in $\calS$ (resp. $\calP$) is denoted with $SP_\calS(s, t)$ (resp. $SP_\calP(s, t)$).
The rectilinear distance between $s$ and $t$ amid splinegons in $\calS$ (resp. $\calP$) is denoted with $dist_\calS(s, t)$ (resp. $dist_\calP(s, t)$).
Unless specified otherwise, a shortest path is shortest with respect to rectilinear metric and the distance is the shortest distance in rectilinear metric.

Section~\ref{sect:corrhourglasses} describes the corridor and hourglass structures in the decomposition of $\cal{F(S)}$ from \cite{journals/talg/ChenW15}.
Section~\ref{sect:staircasestr} extends the staircase structures for polygonal domains from \cite{journals/comgeo/InkuluK09a} to splinegonal domain.
Appendix provides detailed proofs of these characterizations. 
The reduction algorithms are detailed in Section~\ref{sect:redalgo}.
Algorithm for the case of concave-in splinegon obstacles is given in Subsection~\ref{subsect:algoconcavein} and the algorithm for the case of arbitrary splinegon obstacles is given in Subsection~\ref{subsect:algoslinegon}.
Conclusions are in Section~\ref{sect:conclu}.

\section{Corridors and hourglasses in splinegonal domain}
\label{sect:corrhourglasses}

In the following, for convenience, we assume that no splinegon in $\calS$ has an edge that is parallel to either of the coordinate axes. 
We first detail the corridor structure that results from applying the algorithm from \cite{journals/talg/ChenW15}.
In the context of polygonal domains, these structures were first described in Kapoor et al.,~\cite{conf/socg/Kapoor88,journals/dcg/KapoorMM97}. 
Chen and Wang in \cite{journals/talg/ChenW15} extended them to splinegonal domain.

In decomposing $\cal{F(S)}$ into corridors, points $s$ and $t$ are considered as two special splinegons in $\calS$.
\cite{journals/talg/ChenW15} first decomposes $\cal{F(S)}$ into $O(n)$ bounded degree regions by introducing $O(n)$ non-intersecting diagonals.
This decomposition of $\cal{F(S)}$ is termed {\it bounded degree decomposition}, denoted with $BDD(\cal{F(S)})$.
In $BDD(\cal{F(S)})$, two regions are neighboring if they share a diagonal on their boundaries.
Each such region has at most four sides and each side is either a diagonal or (part of) a splinegon edge and has at most three neighboring regions.
In addition to the splinegon vertices, the endpoints of the diagonals of $BDD(\cal{F(S)})$ are also treated as the vertices of $BDD(\cal{F(S)})$. 
Let $G(\cal{F(S)})$ denote the planar dual graph of $BDD(\cal{F(S)})$.
Since each region in $BDD(\cal{F(S)})$ has at most three neighbors, $G(\cal{F(S)})$ is a planar graph whose vertex degrees are at most three.

\begin{wrapfigure}{r}{3cm}
\includegraphics[width=3cm]{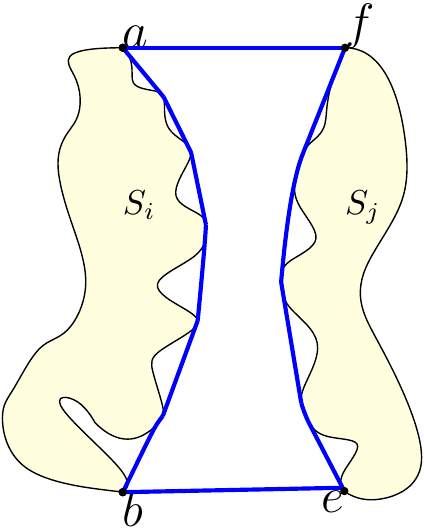}
\caption{Illustrating an open hourglass (blue).}
\label{fig:f1}
\end{wrapfigure}

Based on $G(\cal{F(S)})$, \cite{journals/talg/ChenW15} computes a planar 3-regular graph, denoted by $G^3$ (the degree of each node in it is three), possibly with loops and multi-edges, as follows. 
First, it removes every degree-one node from $G(\cal{F(S)})$ along with its incident edge; repeats this process until no degree-one node exists. 
Second, the algorithm removes every degree-two node from $G(\cal{F(S)})$ and replaces its two incident edges by a single edge; it repeats this process until no degree-two node exists. 
The number of faces, nodes, and edges in the resulting graph $G^3$ is proved to be $O(h)$. 
Each node of $G^3$ corresponds to a region of $BDD(\cal{F(S)})$, which is called a junction region. 
Removal of all junction regions from $G^3$ results in $O(h)$ corridors, each of which corresponds to one edge of $G^3$.



The boundary of each corridor $C$ consists of four parts (refer Figs.~\ref{fig:f1},~\ref{fig:f2}): 
(1) A boundary portion of a splinegon obstacle $S_i \in \calS$, from a point $a$ to a point $b$,
(2) a diagonal of a junction triangle from $b$ to a point $e$ on an obstacle $S_j\in \calS$ ($S_i=S_j$ is possible),
(3) a boundary portion of the obstacle $S_j$ from $e$ to a point $f$, and
(4) a diagonal of a junction triangle from $f$ to $a$.

\begin{wrapfigure}{r}{3.5cm}
\includegraphics[width=3.5cm]{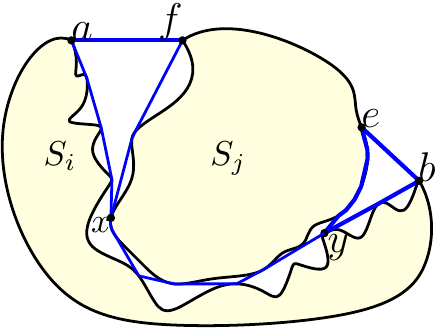}
\caption{Illustrating a closed hourglass (blue).}
\label{fig:f2}
\end{wrapfigure}

Note that the corridor $C$ itself is a simple splinegon.
Let $\pi(a,b)$ (resp., $\pi(e,f)$) be the Euclidean shortest path from $a$ to $b$ (resp., $e$ to $f$) in $C$. 
The region $H_C$ bounded by $\pi(a,b), \pi(e,f)$, $\overline{be}$, and $\overline{fa}$ is called an {\it hourglass}, which is {\it open} if $\pi(a,b)\cap \pi(e,f)=\emptyset$ and {\it closed} otherwise. 
If $H_C$ is open (refer Fig.~\ref{fig:f1}.), then both $\pi(a,b)$ and $\pi(e,f)$ are convex chains and are called the {\it sides} of $H_C$.
Otherwise, $H_C$ consists of two {\it funnels} and a path $\pi_C=\pi(a,b)\cap \pi(e,f)$ joining the two apices of the two funnels, and $\pi_C$ is called the {\it corridor path} of $C$ (refer Fig.~\ref{fig:f2}).
The paths $\pi(b,x), \pi(e,x), \pi(a,y)$, and $\pi(f,y)$ are termed {\it sides of funnels} of hourglass $H_C$.
The sides of funnels are convex chains.

In the following subsections, analogous to \cite{journals/comgeo/InkuluK09a}, we define staircase structures for the splinegonal domain and use it in defining the visibility graph of the splinegonal domain.

\section{Staircase structures for a splinegonal domain}
\label{sect:staircasestr}

Following \cite{journals/comgeo/InkuluK09a}, we define staircase structures for the splinegonal domain $\calS$.
Let $p$ and $q$ be points on a convex chain $CC$ which is a side of an hourglass.
Then the boundary between $p$ and $q$ along $CC$ is termed a {\it section of $CC$}. 

The set of vertices $V_{\mathrm{ortho}}$ is defined  such that $v \in V_{\mathrm{ortho}}$ if and only if either of the following is true:
\begin{itemize}
\item[(i)] $v$ is  an endpoint of a corridor convex chain,
\item[(ii)] $v$ is a point on some corridor convex chain $CC$, with the property that there exists either a horizontal or a vertical tangent to $CC$ at $v$.
\end{itemize}

Let ${\cal O}(p)$ be the orthogonal coordinate system defined with $p \in V_{\mathrm{ortho}}$ as the origin, and horizontal $x$-axis and vertical $y$-axis passing through $p$.
We next adopt and redefine the staircase structures from \cite{conf/socg/ClarksonKV87,journals/comgeo/InkuluK09a}. 
For $i \in \{1, 2, 3, 4\}$, we define a set of points $\pi_i(p)$ as: a point $r \in \pi_i(p)$ if and only if $r \in V_{\mathrm{ortho}}$ and $r$ is located in the $i^{th}$ quadrant of ${\cal O}(p)$.  
Furthermore, we define a set of points $S_i(p)$ as follows: 
a point $q$ is in the set $S_i(p)$ (refer Fig.~\ref{fig:s1p}) if and only if

\begin{wrapfigure}{r}{6.5cm}
\includegraphics[width=6.5cm]{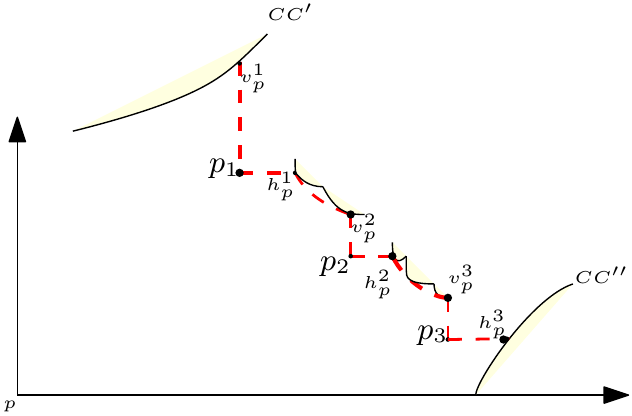}
\caption{Illustrating staircase structure with $S_1(p) = {p_1,p_2,p_3}$.}\label{fig:s1p}
\end{wrapfigure}

\begin{itemize}
\item[(i)]
 $q \in \pi_i(p)$, 
\item[(ii)]
 there is no $p'$ (distinct from $p$) such that $p'$ is in $\pi_i(p)$ and $q $ is in $\pi_i(p')$, and
\item[(iii)]
 $q$ is visible from $p$.
\end{itemize}

We assume that the points in $S_i(p)$ are sorted in increasing x-order.

The proofs of the lemmas and theorems stated in this section are similar to the ones provided for the polygonal domain in \cite{journals/comgeo/InkuluK09a}.
For the sake of completion, we have provided the proofs of these in Appendix. 

\begin{lemma}
\label{lem:ascseq}
Sorting the set of points in $S_1(p)$ in increasing x-order results in the same set of points being sorted in decreasing y-order (or, vice versa).
\end{lemma}

We term two points $\{p_u, p_v\} \subseteq S_i(p)$  as {\it adjacent} in $S_i(p)$ if no point $p_l \in S_i(p)$ occurs between $p_u$ and $p_v$ when the points in $S_i(p)$ are ordered  by either the x- or $y$-coordinates.

Let $p_1, p_2, \ldots, p_k$ be the points in $S_1(p)$ in increasing $x$-order.
Let $h_j$ be the rightward horizontal ray from $p_j$.
And, let $v_j$ be the upward vertical ray from $p_j$.
The ray  $h_j$ intersects either a corridor convex chain or $v_{j+1}$.
Let this point of intersection be $h_{j}^p$.
The ray $v_j$ first intersects  either a corridor convex chain or $h_{j-1}$.
Let this point of intersection be $v_{j}^p$.
If the ray does not intersect any other line or line segment then the point $h_j^p$ or $v_j^p$ could be at infinity.
Let $R_j$ ($j \in \{1, \ldots, k\}$) denote the unique sequence of sections of corridor convex chains/bounding edges that join $h_{j}^p$ and $v_{j+1}^p$.
As will be proved, $R_j$ is continuous.
Note that for the case in which $h_{j}^p = v_{j+1}^p$, $R_j$ is empty.
The elements in the set $\bigcup_{\forall j \in \{1,2,\ldots,k\}} (v_j \cup h_j \cup R_j)$ form a contiguous sequence, termed as the {\em $S_1(p)$-staircase} (refer Fig.~\ref{fig:s1p}).
Analogously, $S_i(p)$ for $i \in \{2,3,4\}$ are defined.
Note that the convex chains which may possibly intersect the coordinate axes and do not contain a point in $S_i(p)$ are not defined to be part of the staircases in the $i$-th quadrant of ${\cal O}(p)$.

We next characterize the structure of a staircase in splinegonal domain.
This is detailed in the following theorem whose proof is given in Appendix.


\begin{lemma}
\label{lem:geoment}
Along the $S_1(p)$-staircase, any two adjacent points in $S_1(p)$ are joined by at most three geometric entities.
These entities ordered by increasing $x$-coordinates are : (i) a horizontal line segment, (ii) a section of a convex chain where tangent to each point in that section of splinegon has a negative slope, and (iii) a vertical line segment.
\end{lemma}

We now define the weighted restricted  visibility graph $G_{\mathrm{vistmp}}$($V_{\mathrm{vistmp}} = V_{\mathrm{ortho}} \cup V_1, E_{\mathrm{vistmp}} = E_{\mathrm{occ}} \cup E_{\mathrm{1}} \cup E_{\mathrm{tmp}}$):
\begin{itemize}
\item[-] For each $v \in V_{\mathrm{ortho}}$, let $v_L$ (resp. $v_R$) be the horizontal projection of $v$ onto the first corridor convex chain in leftward (resp. rightward) direction.
Similarly, let $v_U$ (resp. $v_D$) be the vertical projection of $v$ onto the first corridor convex chain in upward (resp. downward) direction.
If no such corridor chain is encountered then the projection occurs at infinity.
(Note that the orthogonal projection does not exist in at least one direction.)
Similarly, let the vertical projections of $v$ in increasing and decreasing direction of $y$-coordinates be $v_U$ and $v_D$ respectively.
For each point  $p \in \{v_L, v_R, v_D, v_U\}$, if the distance of $p$ from $v$ is finite  then $p$ is added to $V_1$ and the edge $pv$ is added to $E_1$.  
The weight of edge $e \in E_1$ is the rectilinear distance between its two endpoints.
\item[-] An edge $e=(p,q)$ belongs to $E_{\mathrm{occ}}$ if and only if the following conditions hold
(i) $\{p, q\} \subseteq V_{\mathrm{vistmp}}$, (ii) both $p$ and $q$ belong to the same corridor convex chain, and 
(iii) no point in $V_{\mathrm{vistmp}}$ lies between $p$ and $q$ along the chain.
The weight of edge $e$ is the rectilinear distance along the section of convex chain between $p$ and $q$.
\item[-] An edge $e'=(p',q')$ with $p' \in V_{\mathrm{ortho}}$ belongs to $E_{\mathrm{tmp}}$ if and only if $q' \in S_i(p')$.  
The weight of $e'$ is the rectilinear distance along $e'$.
\end{itemize}

\begin{theorem}
\label{thm:vistempcorr}
Let $\{p, q\} \subseteq V_{\mathrm{vistmp}}$.  Then a shortest path from $p$ to $q$ in $G_{\mathrm{vistmp}}$ defines a shortest path in $L_1$ metric from $p$ to $q$ that does not intersect any of the splinegon obstacles in $\calS$.
\end{theorem}

\section{Reduction algorithms}
\label{sect:redalgo}

Given a splinegonal domain $\calS$ and two points $s \in \calF(\calS)$ and $t \in \calF(\calS)$, we reduce the problem of computing a $SP_{\calS}(s,t)$ to the problem of finding a rectilinear shortest path between two points in $\cal{F(P)}$ where $\calP$ is a polygonal domain. 
The $\calP$ comprises of polygons, each of which correspond to a unique splinegon in $\calS$. 
We construct polygons in $\calP$ such that the rectilinear distance between $s$ and $t$ amid polygonal obstacles in $\calP$ is equal to the rectilinear distance between $s$ and $t$ amid splinegons in $\calS$.

We first describe the reduction procedure when the splinegon obstacles in $\calS$ are concave-in.
As mentioned, the reduction for this special case does not relies on the details of the algorithm used to compute a rectilinear shortest path in polygonal domain.   
In a later subsection, we devise an algorithm for the general case in which splinegon obstacles can be arbitrary.

\subsection{Concave-in splinegon obstacles}
\label{subsect:algoconcavein}

  \begin{wrapfigure}{r}{3.4cm}
    \includegraphics[width=3.4cm]{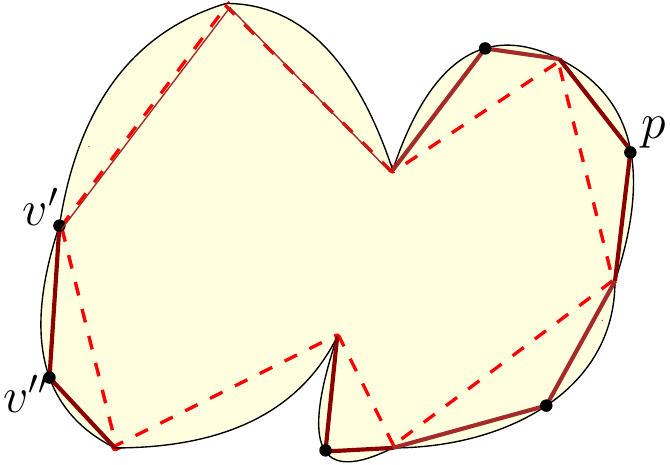}
    \caption{Illustrating carrier polygon (dashed) and the polygon constructed (brown).}
    \label{fig:carrier}
  \end{wrapfigure}

In this subsection, we devise an algorithm to find a shortest path in rectilinear metric when $\calS$ comprises of concave-in splinegons.
We reduce this problem to the problem of computing a shortest path in rectilinear metric amid simple polygonal obstacles.
We compute the set $\calP$ of $h$ simple polygonal obstacles from splinegon obstacles in $\calS$. 
For each splinegon $S \in \calS$ that has $n'$ vertices, we compute its corresponding simple polygonal obstacle $P \in \calP$ with $O(n')$ vertices.
Further, we introduce points $s$ and $t$ in $\cal{F(P)}$ at their respective coordinate locations. 

\ignore {
In computing $\calP$ from $\calS$, we ensure the following
(i) $P \subseteq S$, \\
(ii) for any two points $s, t \in \cal{F(S)}$ that are outside the carrier polygons of splinegons in $\calS$, there is a rectilinear shortest path between $s$ and $t$ amid splinegons in $\calS$ if and only if there exists a rectilinear shortest path between $s$ and $t$ amid polygons in $\calP$ between $s$ and $t$, and \\
(iii) the rectilinear shortest distance between $s$ and $t$ in $\cal{F(S)}$ is equal to the rectilinear shortest distance between $s$ and $t$ in $\cal{F(P)}$.
}

For every $S \in \calS$, we define the vertex set of $P_S \in \calP$ that corresponds to $S \in \calS$ as described herewith: every vertex of $S$ is a vertex of $P_S$; for every point $p \in bd(S)$, if tangent to $S$ at $p$ is either horizontal or vertical then $p$ is a vertex of $P$.
Apart from these two sets of points, no additional point is a vertex of $P_S$.
Let $V_P$ be the set of vertices of $P_S$. 
For any two successive vertices $v', v'' \in V_P$ that occur successively while traversing $bd(S)$, we add an edge between $v'$ and $v''$ to obtain polygon $P_S$ (refer Fig.~\ref{fig:carrier}).

\begin{lemma}
\label{lem:simp}
Every polygon $P \in \calP$ is a simple polygon.
\end{lemma}
\begin{proof}
Let $P_S \in \calP$ be the carrier polygon of $S \in \calS$.
For any vertex $v$ of $P_S$ that is not a vertex of $S$, we know that $v$ is both exterior to the carrier polygon of $S$ and it belongs to a spline $s$ of $S$.
Let $e$ be an edge of carrier polygon of $S$ whose endpoints are same as the endpoints of $s$.
It is immediate that two edges that incident to $v$ intersect only at $v$.
Since $S$ is both simple and concave-in, no point belonging to any edge of $P_S$ lies in the region bounded by $s$ and $e$ (refer Fig.~\ref{fig:carrier}). 
\end{proof}

\begin{lemma}
\label{lem:cvxpolyspline}
If a path $Q \in \cal{F(S)}$ is shortest between $s$ and $t$ amid polygons in $\cal{F(P)}$, then $Q$ is a shortest path amid splinegons in $\calS$.
\end{lemma}
\begin{proof}
Let  $Q'$ be a s-t path amid splinegons in $\calS$ which is shorter than $Q$. 
Then $Q'$  belongs to $\cal{F(S)}$.
Since $\cal{F(P)} \subseteq \cal{F(S)}$, then $Q'$ should also belong to $\calF(P)$.
\end{proof}

\vspace{-0.2in}

\begin{figure}
\center{\includegraphics[width=5cm]{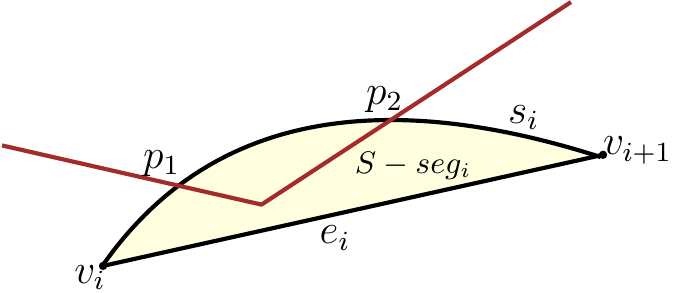}}
\caption{Illustrating the case of a path entering $S$-$seg_i$ region.}
\label{fig:entry}
\end{figure}

\vspace{-0.3in}

\begin{lemma}
\label{lem:ssegonce}
There exists a rectilinear shortest path $Q$ between $s$ and $t$ amid polygonal obstacles in $\calP$ such that no point of $Q$ belongs to any of the open $S$-$seg$ regions of splinegons in $\calS$. 
\end{lemma}
\begin{proof}
Let $Q$ be a shortest path that enters $S$-$seg_i$ region at a point $p_1$ and exits it at $p_2$.
(Refer to Fig.~\ref{fig:entry}.)
Let $s_i$ be the spline to which $p_1$ and $p_2$ belong. 
Also, let $s_i$ bounds a side of $S$-$seg_i$.
Since we have introduced vertices of $P_S$ at every point on the boundary of splinegon where there is a horizontal and/or vertical tangent to splinegon, spline $s_i$ is $xy$-monotone.
Hence, replacing the simple path $Q$ from $p_1$ to $p_2$ with the section of $s_i$ from $p_1$ to $p_2$ does not increase the length of $Q$. 
\end{proof}

\begin{lemma}
\label{lem:dist}
The rectilinear distance between $s$ and $t$ amid polygonal obstacles in $\calP$ equals to the rectilinear distance between $s$ and $t$ amid splinegons in $\calS$.
\end{lemma}
\begin{proof}
There are two cases to consider.
Suppose $SP_\calP(s, t)$ does not intersect with any of the open $S$-$seg$ regions.
In this case, since $\cal{F(S)} \subseteq \cal{F(P)}$, $SP_\calP(s,t)$ is a shortest path amid splinegons in $\calS$. 
On the other hand, suppose that $SP_\calP(s, t)$ does intersect with an open $S$-$seg_i$ region. 
From the proof of Lemma~\ref{lem:ssegonce}, we know that there exists a path between $s$ and $t$ that avoids the open $S$-$seg_i$ region.
Let $Q$ be the path between $s$ and $t$ resultant of all such sub-path replacements in $SP_\calP(s, t)$. 
The rectilinear distance along $Q$ equals to the rectilinear distance along $SP_\calP(s, t)$.
Further, $Q$ belongs to $\cal{F(S)}$.
Hence, due to the Lemma~\ref{lem:cvxpolyspline}, $Q$ is a shortest path between $s$ and $t$ amid splinegons in $\cal{S}$.
\end{proof}

\ignore {
\begin{figure}
\center{\includegraphics[width=3cm]{figs/sweep.pdf}}
\caption{Illustrating the sweep line algorithm.}
\label{fig:sweep}
\end{figure}
}

\begin{lemma}
\label{lem:transf}
Computing $\calP$ from $\calS$ takes $O(n)$ time.
\end{lemma}
\begin{proof}
To compute $\calP$ from $\calS$ we need to find a set $T$ comprising of points on edges of splinegons in $\calS$ such that every point $p \in T$ has either a horizontal or a vertical tangent to the spline on which $p$ resides. 
Because of our model of computation, we can find all the $k$ points of a spline $S$ that belong to $T$ in $O(k)$ time.
Since there are $n$ edges in $\calS$ and each edge has $O(1)$ points that belong to $T$, there are $O(n)$ vertices that define $\calP$.
Including the cost of traversal of each spline to compute polygons in $\calP$, algorithm takes $O(n)$ time to compute $\calP$.
\end{proof}

To transform  $SP_\calP(s, t)$ to $SP_\calS(s, t)$, a plane sweep algorithm is used to find the points of intersection of the $SP_\calP(s,t)$ with the splinegons in $\calS$. 
We sort the endpoints of the line segments in $SP_\calP(s, t)$ with respect to their $y$-coordinates. 
For every splinegon $S_i \in \calS$, let $S_i^{max}$ (resp. $S_i^{min}$) be a point on the boundary of $S_i$ that has the largest (resp. smallest) $y$-coordinate among all the points of $S_i$.
We sort all the points in the set $T$ comprising of $\bigcup_i (S_i^{max} \cup S_i^{min})$.
We use balanced binary search trees to respectively implement the event queue and the status structure needed for plane sweep.
The left to right order of the segments along the sweep line corresponds to the left to right order of the leaves in the balanced binary search tree (status structure).
We sweep the plane with a horizontal line from the point that has the maximum coordinate in $T$ to the point that has the minimum coordinate in $T$.
Let $p$ be an endpoint of the line segment $e \in SP_\calP(s,t)$.
When $p$ is encountered by the sweep line, we check if there is a splinegon, say $S$, immediately to the right or left of the edge $e$ in the status structure; if $S$ exists, we find the points of intersection of $e$ with the $S$ using the algorithm given in \cite{journals/algorithmica/DobkinSW88}.

\begin{lemma}
\label{lem:transf2}
Computing a shortest path between $s$ and $t$ amid splinegons in $\calS$ takes $O((h+k)\lg{n}+(h+k+k')\lg{h+k})$ time, where $h$ is the number of obstacles, $k$ is the number of line segments of $SP_\calP(s, t)$, and $k'$ is the number of intersection points of $SP_\calP(s, t)$ with the splinegon obstacles in $\calS$.
\end{lemma}
\begin{proof}
If there is a splinegon $S$ immediately to the left or right of a line segment $l$ of $SP_\calP(s,t)$, then we can find the intersection of $l$ with $S$ in $O(\lg{n'})$ time using the algorithm given in \cite{journals/algorithmica/DobkinSW88}, where $n'$ is the number of vertices of $S$.
Computing and sorting the event points take $O((h+k)\lg{(h+k)})$ time. 
We check whether a line segment of $SP_\calP(s, t)$ intersects a splinegon when the sweep line reaches endpoints of segments of $SP_\calP(s, t)$ or when it encounters points that belong to $T$; and the number of these event points is $O(h+k)$.
Since to check the points of intersection at each event point requires $O(\lg{n})$ time, the total time required to find the points of intersection at event points take $O((h+k)\lg{n})$ time. 
We update the status structure whenever the sweep line encounters either of these points: intersection points of line segments of $SP_\calP(s, t)$ with splinegons in $\calS$; points belonging to $T$; endpoints of line segments of $SP_\calP(s,t)$.
Considering that updating the status structure per one such event takes $O(\lg{(h+k)})$ time, the time required for all updates together is $O((h+k)\lg{n}(h+k+k')\lg{(h+k)})$.
\end{proof}

\begin{theorem}
Given a splinegonal domain $\calS$ comprising of $h$ pairwise disjoint simple concave-in splinegons together defined with $n$ vertices and two points $s, t \in \cal{F(S)}$, the reduction procedure to compute a rectilinear shortest path between $s$ and $t$ amid splinegons in $\calS$,
excluding the time to compute a rectilinear shortest path amid polygons in $\calP$,
takes $O(n+(h+k)\lg{n}+(h+k+k')\lg{h+k})$ time. 
Here, $\calP$ is the computed polygonal domain from $\calS$, $k$ is the number of line segments in the polygonal shortest path $SP_\calP(s, t)$ between $s$ and $t$ amid polygonal obstacles in $\calP$, and $k'$ is the number of points of intersections of $SP_\calP(s, t)$ with splinegons in $\calS$.
\end{theorem}
\ignore {
\begin{proof}
The correctness of the shortest path between $s$ and $t$ amid $\calS$ follows from Lemma~\ref{lem:cvxpolyspline} and Lemma~\ref{lem:dist}
The time complexity follows from Lemma~\ref{lem:transf} and Lemma~\ref{lem:transf2}. 
\end{proof}
}

\subsection{Simple splinegon obstacles}
\label{subsect:algoslinegon}

We first describe the algorithm when the decomposition of $\cal{F(S)}$ has only open corridors. 
Later we extend this algorithm to handle the closed corridors.

\begin{figure}[h]
\begin{center}
\includegraphics[width=0.65\textwidth]{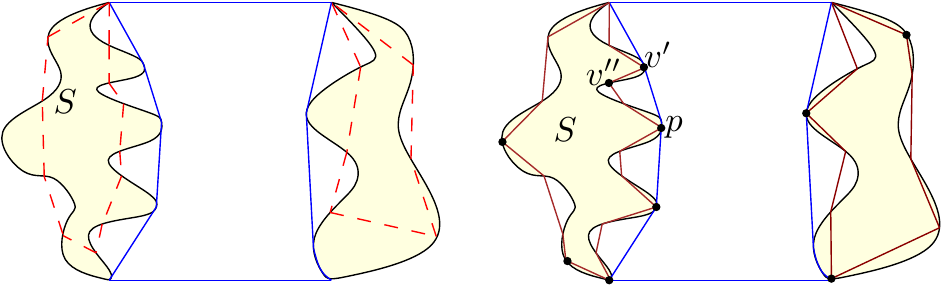}
\caption{Illustrating two splinegons whose sections of boundary belong to an open corridor: splinegons are in black, carrier polygon is red (left), hourglass in blue, and polygons belonging to $\calP$ are in brown (right).}
\label{fig:open}
\end{center}
\end{figure}

\ignore {
\begin{wrapfigure}{r}{2cm}
    \includegraphics[width=2cm]{figs/proj.pdf}
    \caption{Illustrating projection of $v\in V_{ortho}$ on a convex chain.}
    \label{fig:proj}
  \end{wrapfigure}
}

In the following, we define a set $V_\calP$ of points; these are used in defining the polygonal domain $\calP$ corresponding to $\calS$.
For every edge $s$ of a splinegon $S \in \calS$, the endpoints of $s$ belong to $V_\calP$; further if a tangent to $s$ at a point $p \in S$ is either horizontal or vertical and $p$ does not lie inside the convex hull of the carrier polygon of $S$, then $p$ is included in $V_\calP$. 
Also, for every side $s'$ of every open hourglass $H_C$ in the splinegonal domain $\calS$, the endpoints of $s$ belong to $V_\calP$. 
Further, for every vertex $v \in V_{ortho}$ in $\calS$ (see Section~\ref{sect:staircasestr} for the definition of $V_{ortho}$), the horizontal and vertical projections of $v$ onto sides of hourglasses are added to $\calP$.
For every splinegon $S \in \calS$, for any two vertices $v', v'' \in V_\calP$ that occur successively while traversing the boundary of $S$, we add an edge between $v'$ and $v''$ to obtain $P_S \in \calP$ corresponding to $S$. 
Further, we introduce points $s$ and $t$ in $\cal{F(P)}$ at the same coordinate locations as they are in $\cal{F(S)}$. 

Let $\calS'$ be the set comprising of convex hulls corresponding to each of the carrier polygons of splinegons in $\calS$.
Since carrier polygons are simple polygons and since no point of $V_\calP$ belongs to the interior of any of the convex hulls in $\calS'$, every polygon in $\calP$ is guaranteed to be a simple polygon.
(Refer to Fig.~\ref{fig:open}.)

\begin{lemma}
\label{lem:simppoly2}
Every polygon $P \in \calP$ is a simple polygon.
\end{lemma}

\begin{lemma}
\label{lem:graph}
If $Q \in \cal{F(P)}$ is a rectilinear shortest path between $s$ and $t$ amid polygons in $\calP$ computed using the algorithm given in \cite{journals/comgeo/InkuluK09a}, then $Q$ is a shortest path between $s$ and $t$ amid splinegons in $\calS$.
\end{lemma}

\ignore{
\begin{figure}[h]
\centering

\includegraphics[width=0.43\linewidth]{figs/nonconvex.pdf}

\caption{Illustrating corridor convex chains (in red) of a splinegon (left) and its corresponding polygon (right).  The chain on the polygon (right) is not a path of shortest length between its endpoints due to $v$.}
\label{fig:nonconvex}
\end{figure}
}

\begin{proof}
Consider the graph $G_{\calP}$ from which a rectilinear shortest path amid simple polygons in $\calP$ between $s$ and $t$ is computed in \cite{journals/comgeo/InkuluK09a}. 
Let $G_{\calS}$ be the graph corresponding to $\calS$, as defined in Section~\ref{sect:staircasestr}.
We prove that $G_{\calP}$ is same as $G_{\calS}$.

The definitions of $V_{ortho}, V_1$ from Section~\ref{sect:staircasestr} are used in the following.
Analogously, we define $V'_{ortho}, V'_1$ for $\calP$ which respectively correspond to $V_{ortho}, V_1$.
Let $V = V_{ortho} \cup V_1$ be the vertex set of $G_{\calS}$ and $V' = V'_{ortho} \cup V'_1$ be the vertex set of $G_{\calP}$. 
We prove that a vertex belongs to $V$ if and only if it belongs to $V'$.
Suppose $v \in V'_{ortho}$ but $v$ does not belong to $V_{ortho}$. 
Then it must be the case that $v$ is hidden by a convex chain $ab$ in splinegonal domain.
Since $a$ and $b$ are endpoints of an hourglass side in the decomposition of $\cal{F(S)}$, these two are vertices of polygons in $\calP$. 
Suppose $v$ is an endpoint of an hourglass side in $\calP$. 
Then this would lead to a contradiction as we could extend the convex chain to $a$ or $b$ in $\calP$.
Suppose $v$ is residing on an hourglass side $ab$ in $\calP$ but not an endpoint of hourglass.
Since the hourglass side is the shortest path between $a$ and $b$ in $\calP$ and since every vertex of the chain lies on the boundary of a splinegon, $v$ being hidden by the convex chain $ab$ would contradict the fact that the chain from $a$ to $b$ is the shortest path between $a$ and $b$ in $\calP$. 
Thus $v$ lies on a convex chain in $\calS$ and $v$ does not lie inside the convex hull of the carrier polygon as $v$ is part of the shortest  path between $a$ and $b$.
Since there is a horizontal (resp. vertical) tangent to $v$ in $\calP$, there exists a horizontal (resp. vertical) tangent at $v$ to a splinegon in $\calS$. 
This contradicts our assumption that $v$ does not belong to $V_{ortho}$, therefore if $v\in V'_{ortho}$ then $v \in V_{ortho}$.
Analogously we can prove the converse.

By the way we defined $V'_1$, it is immediate to note that a vertex $v \in V_1$ if and only if $v \in V'_1$. \par
Now we show that for every edge $e' \in G_{\calP}$ we introduce a corresponding edge $e \in G_{\calS}$ such that the weights of the corresponding edges are same.
Let $e$ be an edge in $G_s$ of length $l$. 
Also let $p$ and $q$ be the endpoints of $e$. 
We prove that there is a path of length $l$ between $p$ and $q$ amid polygonal obstacles in $\calP$ as well.
The definitions of $E_1$, $E_{occ}$ and $E_{tmp}$ are given in Section~\ref{sect:staircasestr}.

\begin{itemize}
\item[-] 
Suppose $e \in E_1$.
Here $l$ is the rectilinear distance between $p\in V_{ortho}$ and $q \in V_1$.
Since we had proven that if a vertex belongs to $V_{ortho}$ (resp. $V_1$) in $\calS$ then the vertex also belongs to $V'_{ortho}$ (resp. $V'_1$) in $\calP$. Thus the rectilinear distance between $p$ and $q$ will be same in both $\cal{F(P)}$ as well as in $\cal{F(S)}$.

\item[-] 
Suppose $e \in E_{occ}$. 
Here $l$ is the rectilinear distance along the (splinegonal) convex chain between $p$ and $q$, where $p$ and $q$ are the consecutive points on the side of an hourglass obtained due to the decomposition of $\cal{F(S)}$.
Since every section of convex chain in the splinegon domain is $xy$ monotone, the rectilinear distance between $p$ and $q$ in $\cal{F(P)}$ equals to the rectilinear distance between $p$ and $q$ in $\cal{F(S)}$.

\item[-]
Suppose $e \in E_{tmp}$.
Here $l$ is the rectilinear distance between $p \in V_{ortho}$ and let $q \in S_1(p)$.
We prove that if $q\in S_1(p)$ in splinegonal domain then $q\in S_1'(p)$ in polygonal domain as well.
Suppose $q$ does not belong to $S_1'(p)$.
Then $q$ is not visible from $p$ amid polygons in $\calP$.
This means that a convex chain of an open hourglass of the decomposition of $\cal{F(S)}$ intersects the line segment $pq$.
However, since the convex chain in polygonal domain is always bounded by a convex chain in $\calS$, this would imply $pq$ is intersected by a convex chain in splinegonal domain as well.
\end{itemize}

Therefore, if $Q$ is a shortest $s$-$t$  path obtained from $G_{\calP}$ then it is also the shortest $s$-$t$ path in $G_{\calS}$.
This together with Theorem~\ref{thm:vistempcorr} lead to conclude that $Q$ is also a shortest path amid splinegons in $\calS$. 
\end{proof}

To find the horizontal and vertical projections of points in $V_{ortho}$, we use the plane sweep algorithm from \cite{journals/comgeo/InkuluK09a} extended to splinegons.
Essentially, we sweep a vertical line from left to right to find the horizontal rightward projections of every $v \in V_{ortho}$.
The status of the vertical sweep line is maintained as a set of points in $V_{ortho}$ that lie on the sweep line, sorted by their $y$-coordinates.
Let $p$ be the first point of a convex chain $CC$ struck by sweep line and let
 $r$ be a point in the sweep-line status structure at the time $p$ is encountered by the sweep line.
If $p$ projects onto $CC$ at $p'$ then $p'$ is a projection of $p$. 
After finding $p'$, we remove $p$ from the sweep-line status structure.
Analogously, projections of points in $V_{ortho}$ are determined.

\begin{lemma}
\label{lem:s2p}
When the free space of the given splinegonal domain $\calS$ is partitioned into open corridors, computing a polygonal domain $\calP$ from $\calS$, so that the $dist_\calP(s, t)$ equals to the $dist_\calS(s, t)$ for two given points $s, t \in \cal{F(S)}$, takes $O(n + h \lg{n})$ time.
\end{lemma}
\begin{proof}
Since there are $n$ edges in $\calS$ and we are adding a constant number of points to each edge, computing $\calP$ from $\calS$ takes $O(n)$ time. 
Each vertex in $V_{ortho}$ is inserted into (resp. deleted from) sweep line data structures' only once, together taking $O(h\lg{h})$ time.
With binary search, intersection of a horizontal (resp. vertical) line from a point with a convex chain can be found in $O(\lg{n})$ time.
Hence all the points of projections can be computed as stated.
\end{proof}

\begin{figure}[h]
\begin{center}
\includegraphics[width=0.4\textwidth]{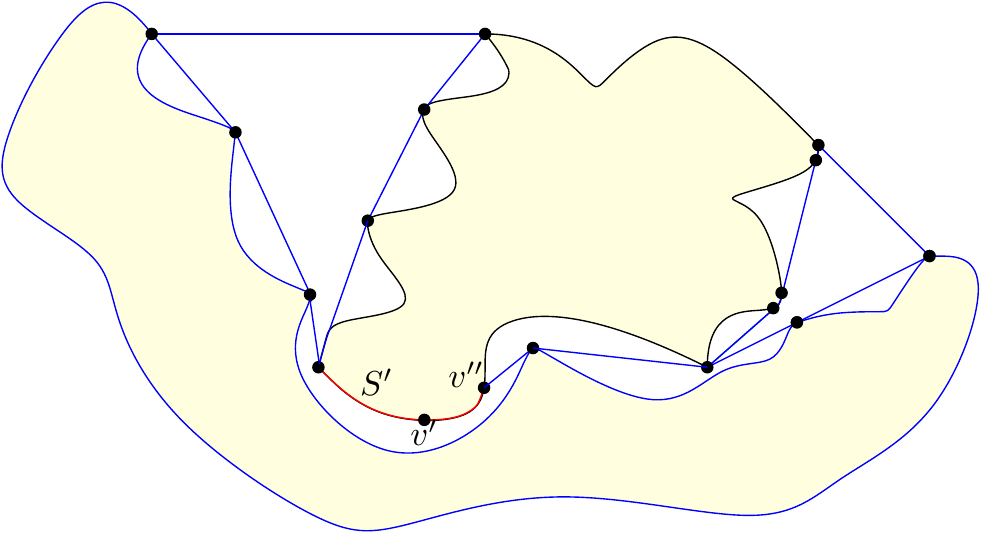}
\caption{Illustrating sections of boundaries of two splinegons participating in a closed corridor and their carrier polygon (red) and sections of polygons computed (brown).}
\label{fig:closed}
\end{center}
\end{figure}

Now we extend this algorithm to handle closed corridors.
For each side of every funnel, very similar to sides of open hourglasses, we introduce points into $V_\calP$; this include projections of points $V_{ortho}$ onto sides of funnels.
Let $Q$ be the rectilinear shortest path between apices of a closed corridor in $\calS$.
For every contiguous maximal section $S'$ of every spline that belong to $Q$, we add the endpoints of $S'$ to $V_{\calP}$.
For every two vertices $v', v' \in V_{\calP}$ if they occur successively along the boundary of a splinegon, then $v'v''$ is an edge of a polygon in $\calP$ (refer Fig.~\ref{fig:closed}).
Note that for any two splinegons that participate in a closed corridor, their corresponding polygons in $\calP$ are guaranteed to be disjoint.

\begin{lemma}
\label{lem:length}
If $a$ and $b$ are the apices of a closed corridor and the shortest distance between them is $d$ in $\cal{F(S)}$ then the shortest distance between $a$ and $b$ in $\cal{F(P)}$ is $d$.
\end{lemma}
\begin{proof}
Suppose there is a path $Q$ in $\calP$ whose length is less than $d$. 
Now $Q$ must intersect with a splinegon $S$.
Let $v'$ and $v''$ be the points of intersection of $Q$ with $S$.
But $Q$ can only intersect with an edge of splinegon which is outside the convex hull of its corresponding carrier polygon.
Since these edges are $xy$-monotone, we can replace the section of path $Q$ between $v'$ and $v''$ with the section of spline edge between $v'$ and $v''$.
\end{proof}

\begin{lemma}
\label{lem:clsd}
If path $Q$ in $\cal{F(S)}$ is a shortest path between $s$ and $t$ amid polygons in $\cal{F(P)}$ then $Q$ is a shortest path between $s$ and $t$ amid splinegons in $\calS$.
\end{lemma}
\begin{proof}
We prove that the graph $G_\calS$ corresponding to $\calS$ is same as the graph $G_\calP$ corresponding to $\calP$.
The only new edges that would occur in $G_\calS$ due to closed corridors are the edges joining two apices of every closed corridor. 
For every such edge $e$, the length $d$ of $e$ is the rectilinear distance between the apices of two funnels of a closed corridor $C$ in $\cal{F(S)}$.
Since the distance between apices of the closed corridor corresponding to $C$ in $\cal{F(P)}$ equals to $d$ and because of Lemma~\ref{lem:length}, both $G_{\calS}$ and $G_{\calP}$ are precisely same. 
\end{proof}

\begin{lemma}
\label{lem:closed}
Computing a polygonal domain $\calP$ from the given splinegonal domain $\calS$ takes $O(n + h \lg{n})$ time.
\end{lemma}
\begin{proof}
Since funnels are processed analogous to open hourglasses and due to Lemma~\ref{lem:s2p}, computing polygonal chains corresponding to sides of open corridors and funnels together take $O(n + h \lg{n})$ time.  
Computing a shortest path between apices of any closed corridor $C$ in $\cal{F(S)}$ takes $O(k)$ time \cite{journals/algorithmica/DobkinS90,journals/siamcomp/MelissaratosS92}, where $k$ is the number of vertices that belong to that closed corridor. 
Further, traversing along a shortest path $Q$ between two apices of $C$ and introducing vertices of $\calP$ along $Q$ takes $O(k)$ time. 
\end{proof}

\begin{theorem}
Given a splinegonal domain $\calS$ comprising of $h$ pairwise disjoint simple splinegons together defined with $n$ vertices and two points $s, t \in \cal{F(S)}$, the reduction procedure to compute a rectilinear shortest path between $s$ and $t$ amid splinegons in $\calS$, excluding the time to compute a rectilinear shortest path amid polygons in $\calP$, takes $O(n+h \lg{n})$ time.
Here, $\calP$ is the polygonal domain computed from $\calS$.
\end{theorem}
\ignore {
\begin{proof}
Lemma~\ref{lem:graph}, Lemma~\ref{lem:length} and Lemma~\ref{lem:clsd} ensures the correctness.
The time complexity follows from Lemma~\ref{lem:s2p} and Lemma~\ref{lem:closed}.
\end{proof}
}

\section{Conclusions}
\label{sect:conclu}

We have devised an algorithm to reduce the problem of computing a rectilinear shortest path between two points in the splinegonal domain to the problem of computing a rectilinear shortest path between two points in the polygonal domain.
The reduction algorithm given for the case of concave-in splinegon obstacles does not rely on details of the algorithm to compute a rectilinear shortest path between two points amid polygonal obstacles.
Further, as part of this, we have generalized few of the properties given for rectilinear shortest paths in polygonal domain to the case of rectilinear shortest paths amid splinegons.
It would be interesting to devise rectilinear shortest path algorithms when the obstacles in plane are more generic.

\bibliographystyle{plain}


\section{Appendix}
\label{sect:append}

\begin{customlemma}{\ref{lem:ascseq}}
Sorting the set of points in $S_1(p)$ in increasing x-order results in the same set of points being sorted in decreasing y-order (or, vice versa).
\end{customlemma}
\vspace*{-.10in}
\begin{proof}
Let $p_1, p_2, \ldots, p_l$ be the points in $S_1(p)$ in increasing $x$-order.
The $y$-coordinate of $p_{j+1}$ cannot be greater than the $y$-coordinate of $p_j$ (for $j \in \{1, \ldots, l\}$) without violating condition (ii) in the definition  of $S_1(p)$.
\end{proof}

\begin{customlemma}{\ref{lem:geoment}}
Along the $S_1(p)$-staircase, any two adjacent points in $S_1(p)$ are joined by at most three geometric entities.
These entities ordered by increasing $x$-coordinates are: (i) a horizontal line segment, (ii) a section of a convex chain where tangent to each point in that section of splinegon has a negative slope, and (iii) a vertical line segment.
\end{customlemma}
\vspace*{-.10in}
\begin{proof}
Let $p_j$ and $p_{j+1}$ be two adjacent points in $S_1(p)$.
Let $l_h$ be the line segment $p_jh^p$, where $h^p=h_j^p$, and $l_v$ be the line segment $p_{j+1}v^p$, where $v^p=v^p_{j+1}$.
Let $R$ be the region bounded by $pp_j, l_h$ and sections of convex chain between $h^p$ and $v^p$ along the staircase, $l_v$ and $p_{j+1}p$.
Let us assume that a section of convex chain $CC$ intersects $R$.

We first note that $CC$ cannot intersect $l_h$ or $l_v$ or the section of convex chain between $h_p$ and $v_p$.
Further, $CC$ cannot intersect $pp_j$ or $pp_{j+1}$ either as it would violate rule (iii) of the definition of $S_1(p)$.
Therefore $CC$ lies completely within $R$.
Then the endpoint of $CC$ with lesser x-coordinate, say $a$, belongs to $S_1(p)$, contradicting our assumption that $p_j$ and $p_{j+1}$ are consecutive points on the staircase.
(Refer Fig.~\ref{fig:int}.)

We prove that if $h^p$ is not the same point as $v^p$ then the two are incident to the same convex chain.
If $h^p$ is equal to $v^p$ then there is nothing to prove.
Let $h^p$ be located on a convex chain $CC_k$ and let $v^p$ be located on a convex chain $CC_l$, where $CC_{k}\neq CC_{l}$.
Let $CC_k, CC_{k+1}, \ldots, CC_{l-1}, CC_{l}$ be the consecutive sequence of sections of convex chains encountered while traversing from $h^p$ to $v^p$.
Let $P$ be the set consisting of points of intersection of any two adjacent entities in this sequence including $h^p$ and $v^p$.
Note that every point of $P$ belongs to the vertex set $V_{ortho}$.
Since $p_j$ and $p_{j+1}$ are adjacent points in $S_1(p)$, we obtain a contradiction if there exists at least one point in $P \cap S_1(p)$ whenever $|P|>2$.
Suppose $P \cap S_1(p)$ is empty and $|P|>2$.
Let $CC_j$ be the first convex chain along the staircase while traversing the staircase in increasing $x$-order, starting at $h^p$. 
Let $p'_j$ be the endpoint of $CC_j$ such that there exists a tangent $pp_t$ to $CC_{j}$ where $p_t$ is located on $CC_{j}$ and $p_t$ is visible to $p$.
If no such convex chain or corresponding point $p_t$ exists then the endpoint of the first convex chain along the staircase (while traversing the staircase from $h_p$) is a point that is visible to $p$ and belongs to $P \cap S_1(p)$.
Thus at least one such $CC_j$ exists.
Let $q_b$ and $q_e$ be the first and last points $CC_j$ as the staircase is traversed from $h^p$ in increasing $x$ coordinates order. 
We prove that there exists a point $r$ located on section of convex chain $CC_j$ between (and including) $q_b$ and $p_t$ such that $r \in S_1(p)$; hence,leading to a contradiction.
Let $e_i$ and $e_{i+1}$ be the sections of $CC_j$ incident on $p_t$. 
The following are the exhaustive cases that arise based on the slopes of $e_i$ and $e_{i+1}$.

\begin{figure}[h]
\centering
\includegraphics[width=0.5\textwidth]{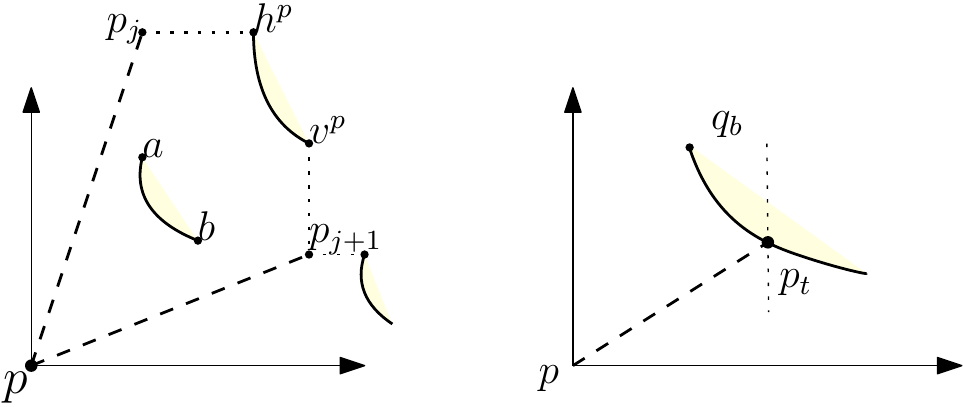}
\caption{Illustrating that no convex chain intersects $R$ (left) and Case 1 (right).}\label{fig:int}
\end{figure}

\begin{itemize}
\item[-] Case 1: W.l.o.g., suppose $e_i$ is below and $e_{i+1}$ is above a horizontal line passing through $p_t$.
In this case both the sections have negative slope at $p_t$.
Then $pp_t$ cannot be a tangent to $CC_j$ as it intersects $CC_j$, contradicting the choice of $CC_j$.
(Refer Fig.~\ref{fig:int}.)

\item[-] Case 2.1: Both $e_i$ and $e_{i+1}$ are of non-negative slope at point $p_t$ and $CC_j$ starts with a negative slope at point $q_b$ with y-coordinate of $q_b$ greater than its adjacent point on $CC_j$.
Then there exists a point $r$ on $CC_j$, located between $q_b$ and $p_t$, such that $r$ has a tangent to $CC_j$ parallel to $x$ axis.
Since $CC_j$ is of negative slope at $q_b$ and $CC_j$ is the first convex chain along the staircase to have a visible point of tangency from $p$, all points located on the section of convex chain between $q_b$ and $p_t$ are visible to $p$.
(Refer Fig.~\ref{fig:case2.2.1}.)

\begin{figure}[h]
\centering
\includegraphics[width=0.5\textwidth]{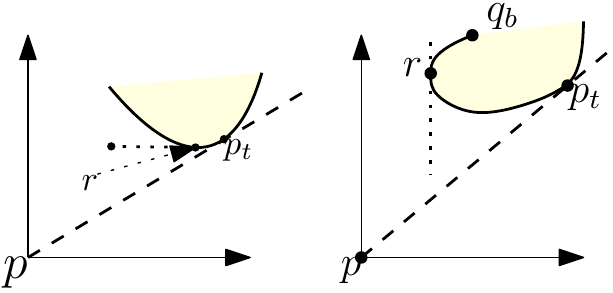}
\caption{Illustrating Case 2.1 (left) and Case 2.2 (right).}\label{fig:case2.2.1}
\end{figure}

\item[-] Case 2.2: Suppose both $e_i$ and $e_{i+1}$ are of non-negative slope at point $p_t$ and $CC_j$ is of non-negative slope at $q_b$.
Also suppose $y$-coordinate of $q_b$ is greater than its adjacent point in $CC_j$.
Then there exists a point $r$ located on $CC_j$ between $q_b$ and $p_t$ such that there is a tangent parallel to $y$-axis on $CC_j$ at $r$.
Further, since $r$ is the leftmost vertex and because of the choice of convex chain $r$ is visible to $p$. 
(Refer Fig.~\ref{fig:case2.2.1}.)

\item[-] Case 2.3: Suppose both $e_i$ and $e_{i+1}$ are of non-negative slope at point $p_t$, and there are two subcases for $CC_j$: (i) $CC_j$ is of non-negative slope; (ii) $CC_j$ is of negative slope at point $q_b$ and the y-coordinate of $q_b$ is less than its adjacent point on $CC_j$.
As $CC_j$ is a convex chain, subcase (i) is possible if and only if at every point from $q_b$ to $p_t$ has a non-negative slope.
This means that if suppose $q_b$ is not an endpoint of the convex chain then there is a convex chain intersecting the region $R$ which results in a contradiction.
(Refer Fig.~\ref{fig:case2.2.3}.)
Next we argue that subcase (ii) never arises.
Consider the section of the convex chain joining $q_b$ with $p_t$.
Since a region that includes obstacle and is bounded by the convex chain must lie outside $R$, the shape of the section of convex chain from $q_b$ to $p_t$ forces that region to be non-convex, contradicting the convex chain properties.
\begin{figure}[h]
\centering
\includegraphics[width=0.5\textwidth]{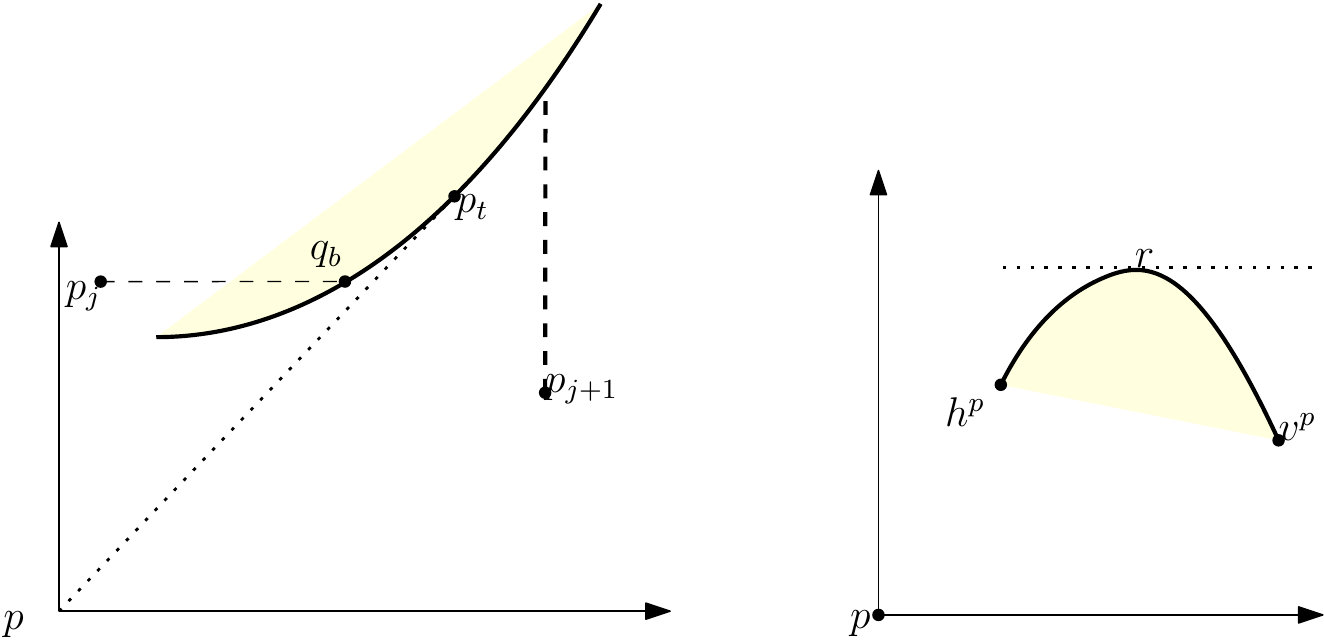}
\caption{Illustrating Case 2.3 subcase (i) and Case (a) (right).}\label{fig:case2.2.3}
\end{figure}
\end{itemize}

Let us name the only possible section of convex chain between $p_j$ and $p_{j+1}$ along the staircase as $CC$.
We prove that the $y$-coordinate of $v_p$ is always less than or equal to the $y$-coordinate of $h_p$, where $v_p$ is the vertical projection of $p_{j+1}$ on $CC$ and $h_p$ is the horizontal projection of $p_j$ on $CC$.
Let us assume that the $y$-coordinate of $v_p$ is greater than $h_p$. 
Since $p_j$ and $p_{j+1}$ are consecutive points in $S_1(p)$, this can only happen if the slope of $CC$ is positive at every point.
However if $h_p$ is not an endpoint of a convex chain then there is a convex chain intersecting the region $R$ which we had proven to be not possible, leading to a contradiction.
Thus the $y$-coordinate of $v_p$ is always less than equal to $h_p$.
Now we prove that every point of $CC$ has a negative slope.

\begin{itemize}
\item[-]
Case (a): Let $CC$ start with a non-negative slope i.e., $h^p$ has a non-negative slope.
We consider two possible subcases, either (i) y-coordinate of $h^p$ is less than that of its adjacent points in $CC$, or (ii) the $y$-coordinate of $h_p$ is greater  than the y-coordinate of its adjacent points in $CC$.
Since the $y$-coordinate of $h_p$ is greater than or equal to the $y$-coordinate of $v^p$, there exists a vertex $r$ in CC where the slope changes from positive to negative.
In subcase (i) the convex chain has an obstacle below $r$ and hence a convex chain intersects $R$, which is a contradiction (refer Fig.~\ref{fig:case2.2.3}).
In the second case, if $r$ is not visible from $p$ then a convex chain intersects $R$, which is a contradiction.
If $r$ is visible from $p$ then $r \in S_1(p)$, which contradicts the assumption that $p_j$ and $p_{j+1}$ are adjacent points along the staircase.

\item[-]
Case (b): Let $CC$ start with a negative slope.
Let $r$ be the point where the slope changes to non-negative.
Thus there is a horizontal/vertical tangent to $CC$  at $r$.
If $r$ is not visible from $p$ then a convex chain intersects $r$. 
However, this is not possible. 
Thus if $r$ is visible to $p$ then $r \in S_1(p)$, which contradicts our assumption that $p_j$ and $p_{j+1}$ are adjacent points in $S_1(p)$.
Thus every point on $CC$ has a negative slope.
\end{itemize}
\end{proof}

\begin{figure}[h]
\centering
\includegraphics[width=0.7\textwidth]{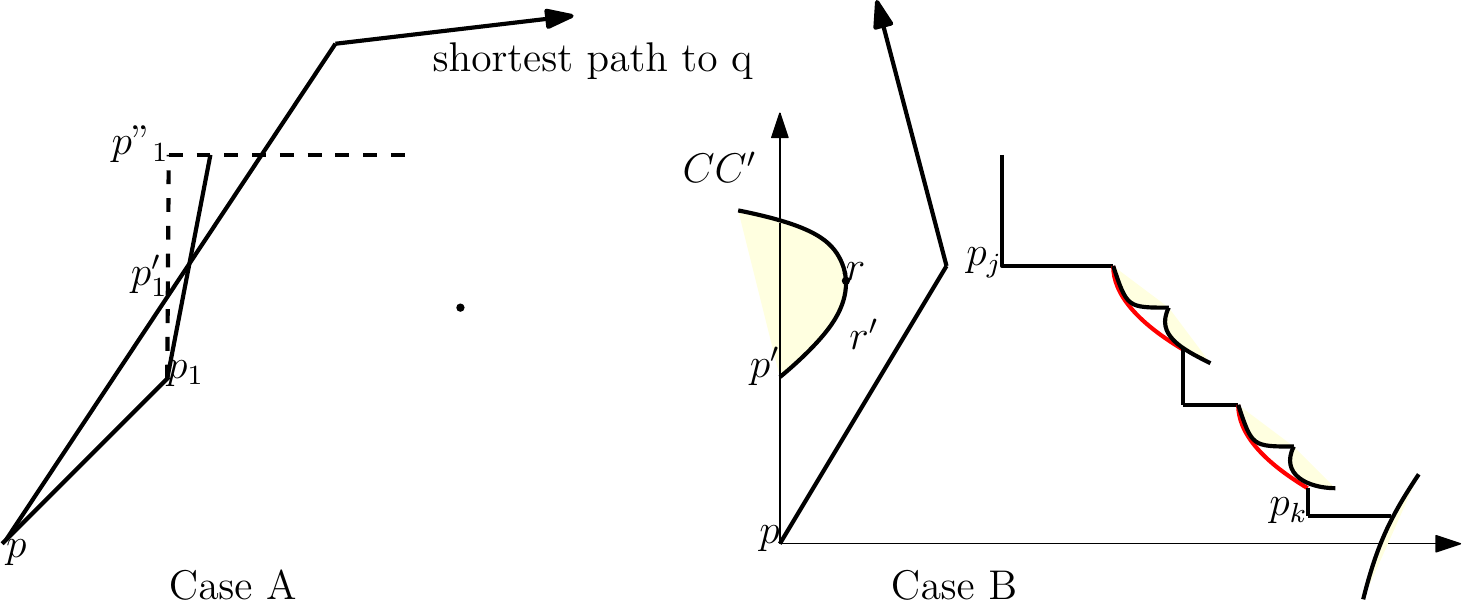}
\caption{Replacing a shortest path from $p$ to $q$ with edges from $G_{\mathrm{vistmp}}$.}
\label{fig:caseb}
\end{figure}

Note similar arguments to Lemma~\ref{lem:geoment} can be given for $S_i(p)$ where $i \in \{2,3,4\}$.

\begin{customthm}{\ref{thm:vistempcorr}}
Let $\{p, q\} \subseteq V_{\mathrm{vistmp}}$.  
Then a shortest path from $p$ to $q$ in $G_{\mathrm{vistmp}}$ defines a rectilinear shortest path from $p$ to $q$ that does not intersect any of the obstacles in $\calS$.
\end{customthm}
\vspace*{-.10in}
\begin{proof}
Consider a shortest path $Q$ from $p$ to $q$ that avoids all obstacles in $calS$. 
We need to consider the following two cases.

Case (a): The shortest path $Q$ crosses a staircase structure defined with respect to point $p$. 
Since convex chains have an obstacle on one side, the shortest path $Q$ does not intersect any of the convex chain in the staircase.
Therefore, the shortest path $Q$ intersects the staircase at either a point in $S_i(p)$, or an orthogonal line segment in the staircase.
Let $p_j$ and $p_k$ be the points in $S_i(p)$ with the minimum and maximum x-coordinates.
For a point $p_1 \in S_i(p)$, suppose the path $Q$ crosses an orthogonal segment $p1p''$  of the staircase at $p_1'$. 
Consider replacing the path from $p$ to $p'_1$ with two line segments: one joining $p$ to $p_1$, and the other from $p_1$ to $p'_1$.
Note that the rectilinear distance along the line joining $p$ to $p'_1$ is same as the rectilinear distance along the altered path.
This new rectilinear path is always guaranteed to exist because of the following reasons: no point of $V_{ortho}$ exists in the region bounded by the staircase and the line segments $pp_j$ and $pp_k$; no convex chain that intersects the coordinate axes intersects with the interior of the altered path.
The path from $p_1$ to $q$ can be altered similarly without changing the length of the path.
Since the distance from $p_1$ to $q$ is shorter than the distance from $p$ to $q$, the process terminates.
Since a shortest path from $p$ to $q$ does not repeat any vertex, the alteration procedure will terminate.
Further, since for $p$ and every $p_l \in S_i(p)$, the edge $pp_l \in E_{vistmp}$, the altered path is in $G_{\mathrm{vistmp}}$.
Therefore, the rectilinear shortest path $Q$ between $p$ and $q$ in the given splinegonal region can be found by determining the shortest path from $p$ to $q$ in the graph $G_{\mathrm{vistmp}}$.

Case (b): The shortest path $Q$ does not intersect any staircase structure defined with respect to $p$. 
Let $pq'$ be the first segment of $Q$ and $pq'$ lie in the first quadrant of ${\cal O}(p)$. 
Let $p_j$ and $p_k$ be the points in $S_1(p)$ with minimum and maximum x-coordinates respectively. 
Since $Q$ does not cross the staircase, it must be the case that that the x-coordinate of $q'$ is either less than $p_j$ or greater than $p_k$.
Consider the first case, if no convex chain intersects the $y$-axis in the first quadrant, then either $q' \in S_1(p)$ or the interior of $pq'$ does not intersect with the first quadrant of ${\cal O}(p)$, leading to a contradiction.
Let $CCY$ be the first convex chain that intersects the upward vertical ray from $p$.
Also let $CC'$ be the section of $CCY$ in the first quadrant of ${\cal O}(p)$ and let $p'$ be the vertical projection of $p$ onto $CC'$.
The slope of $CC'$ at $p'$ will either be positive or negative.
Let $r$ be either the endpoint of $CC'$ or the point where $CC'$ changes the sign of slope.
Therefore $r \in v_{ortho}$ and the $x$-coordinate of $r$ is less than $p_j$.
Let $CC'$ start with negative slope at $p'$.
Then $r$ is visible to $p$ and belongs to $S_1(p)$.
Thus contradicting the choice of $p_j$. 
Therefore $CC'$ starts with a positive slope at $p'$ and $r$ does not belong to $S_1(p)$.
Let $r''$ be the point with smallest $x$-coordinate such that $r'' \in V_{ortho}$ but $r''$ is not visible from $p$ (note that $r$ could be equal to $r''$).
(Refer Fig.~\ref{fig:caseb}.) 
Let $r'$ be the horizontal projection of $r''$ onto $CC'$.
Replace the path from $p$ to $q$ with a path of equal distance so that the latter consists of (i) a vertical line segment $pp'$ (ii) a path from $p'$ to $r'$ along $CC'$, and (iii) a path from $r'$ to $q$.
The rectilinear distance along the $Q$ is same as that of the altered path. 
Since the slopes of points on $CC'$ along the path from $p'$ to $r'$ cannot be negative, the new path is guaranteed to exist. 
The path from $r'$ onwards can similarly be modified.
\end{proof}

\end{document}